\documentclass[11pt]{article}
\usepackage[utf8]{inputenc}
\usepackage[T1]{fontenc}
\usepackage{lmodern}
\usepackage[margin=1in]{geometry}
\usepackage{amsmath,amssymb,amsthm}
\usepackage{graphicx}
\usepackage{booktabs}
\usepackage{array}
\usepackage{xcolor}
\usepackage{microtype}
\usepackage{caption}
\usepackage{subcaption}
\usepackage{float}
\usepackage{flafter}
\usepackage{placeins}
\usepackage[hidelinks]{hyperref}
\usepackage{url}
\newtheorem{lemma}{Lemma}
\newtheorem{proposition}{Proposition}

\providecommand{\RDevice}{Apple M4 Pro}
\providecommand{\RMaxRelErr}{0.00}
\providecommand{\RConvSlope}{-0.50}
\providecommand{\RThroughput}{0.0}
\providecommand{\RCpuGpuAgree}{0.0}
\providecommand{\RErrQuarter}{0}
\providecommand{\RErrHalf}{0}
\providecommand{\RErrOne}{0.0}

\providecommand{\RTauOne}{0.0}
\providecommand{\RTauTenLo}{0.0}
\providecommand{\RTauTenHi}{0.0}
\providecommand{\RTauOneLo}{0.0}
\providecommand{\RTauOneHi}{0.0}
\providecommand{\RTauThin}{0.00}
\providecommand{\RBruteThinDrop}{0}
\providecommand{\RDipOverThin}{0}
\providecommand{\RThickAgree}{0.0}
\providecommand{\RRbruteThick}{0.000}
\providecommand{\RRdipole}{0.000}
\providecommand{\RReflAgree}{0}
\providecommand{\RReflRateBrute}{0.00}
\providecommand{\RReflRateClosed}{0.00}
\providecommand{\RTrAgree}{0}
\providecommand{\RTrRateBrute}{0.00}
\providecommand{\ROptSech}{0}
\providecommand{\ROptCorr}{0.00}
\providecommand{\ROptCorrGeom}{0.00}
\providecommand{\RRimOverGlow}{0.00}
\providecommand{\RCenterRatio}{0.00}
\providecommand{\RRimOverGlowPct}{0}
\providecommand{\RHeroRes}{0}
\providecommand{\RHeroSec}{0.00}
\providecommand{\RHeroWalks}{0}
\providecommand{\RNMedia}{0}
\providecommand{\RHetErr}{0.00}
\providecommand{\RHetCRatio}{2}
\providecommand{\RShadowDrop}{0}
\providecommand{\RSingleNear}{0}
\providecommand{\RSingleTotal}{0}
\providecommand{\RSingleFar}{0.0}
\providecommand{\RSingleCorr}{0.000}
\providecommand{\RInvCentreErr}{0.0}
\providecommand{\RInvContrastErr}{0}
\providecommand{\RInvSrc}{2}
\providecommand{\RInvDet}{0}
\providecommand{\RInvContrastTrue}{0}
\providecommand{\RInvDecay}{0.0}
\providecommand{\RInvResTwo}{0.0}
\providecommand{\RInvResFive}{0.0}
\providecommand{\RCertRho}{0.00}
\providecommand{\RCertDipErr}{0}
\providecommand{\RCertDipWC}{0}
\providecommand{\RCertCorrErr}{0}
\providecommand{\RCertCorrWC}{0}
\providecommand{\RCertOpFrac}{0}
\providecommand{\RCertOpErr}{0}
\providecommand{\RCertOpTol}{0}
\providecommand{\RCertStressShapes}{0}
\providecommand{\RCertStressRho}{0.00}
\providecommand{\RCertStressDipErr}{0}
\providecommand{\RCertStressCorrErr}{0}
\providecommand{\RGtCases}{1}
\providecommand{\RGtErrOurs}{0}
\providecommand{\RGtErrDip}{0}
\providecommand{\RGtErrCorr}{0}
\providecommand{\RGtMeanOurs}{0}
\providecommand{\RGtMeanDip}{0}
\providecommand{\RGtMeanCorr}{0}
\providecommand{\RGtToneOurs}{0.000}
\providecommand{\RGtToneDip}{0.000}
\providecommand{\RGtToneCorr}{0.000}
\providecommand{\RGtToneHybrid}{0.000}

\providecommand{\RGtWinOurs}{0}
\providecommand{\RGtWinDip}{0}
\providecommand{\RGtWinCorr}{0}

\providecommand{\RGtReduce}{0}
\providecommand{\RGtThickX}{0}
\providecommand{\RGtSpp}{0}
\providecommand{\RGtBounce}{0}
\providecommand{\RGtPtSec}{0}
\providecommand{\RGtOursSec}{0}

\providecommand{\RGswShapes}{0}
\providecommand{\RGswFitTau}{0.00}
\providecommand{\RGswFitCurv}{0.00}
\providecommand{\RGswCoefCurv}{0.0}
\providecommand{\RHstNoiseMin}{0}
\providecommand{\RHstNoiseMax}{0}
\providecommand{\RHstRatioMax}{0}
\providecommand{\RHstExpG}{0.0}
\providecommand{\RHstSdfExp}{0.0}
\IfFileExists{numbers.tex}{\renewcommand{\RDevice}{Apple M4 Pro}
\renewcommand{\RMaxRelErr}{0.75}
\renewcommand{\RConvSlope}{-0.49}
\renewcommand{\RThroughput}{4.1}
\renewcommand{\RCpuGpuAgree}{0.5}
\renewcommand{\RErrQuarter}{32}
\renewcommand{\RErrHalf}{14}
\renewcommand{\RErrOne}{3.9}

\renewcommand{\RTauOne}{1.8}
\renewcommand{\RTauTenLo}{0.7}
\renewcommand{\RTauTenHi}{1.0}
\renewcommand{\RTauOneLo}{1.8}
\renewcommand{\RTauOneHi}{2.1}
\renewcommand{\RTauThin}{0.25}
\renewcommand{\RBruteThinDrop}{49}
\renewcommand{\RDipOverThin}{95}
\renewcommand{\RThickAgree}{0.8}
\renewcommand{\RRbruteThick}{0.441}
\renewcommand{\RRdipole}{0.437}
\renewcommand{\RReflAgree}{26}
\renewcommand{\RReflRateBrute}{1.99}
\renewcommand{\RReflRateClosed}{2.01}
\renewcommand{\RTrAgree}{10}
\renewcommand{\RTrRateBrute}{0.99}
\renewcommand{\ROptSech}{5\times10^{-5}}
\renewcommand{\ROptCorr}{0.99}
\renewcommand{\ROptCorrGeom}{0.85}
\renewcommand{\RRimOverGlow}{1.55}
\renewcommand{\RCenterRatio}{1.04}
\renewcommand{\RRimOverGlowPct}{55}
\renewcommand{\RSingleNear}{24}
\renewcommand{\RSingleTotal}{18}
\renewcommand{\RSingleFar}{0.0}
\renewcommand{\RSingleCorr}{0.999}
\renewcommand{\RInvCentreErr}{5.4}
\renewcommand{\RInvContrastErr}{1}
\renewcommand{\RInvSrc}{2}
\renewcommand{\RInvDet}{37}
\renewcommand{\RInvContrastTrue}{8}
\renewcommand{\RInvDecay}{1.2}
\renewcommand{\RInvResTwo}{1.8}
\renewcommand{\RInvResFive}{1.0}
\renewcommand{\RCertRho}{0.77}
\renewcommand{\RCertDipErr}{35}
\renewcommand{\RCertDipWC}{105}
\renewcommand{\RCertCorrErr}{20}
\renewcommand{\RCertCorrWC}{49}
\renewcommand{\RCertOpFrac}{32}
\renewcommand{\RCertOpErr}{10}
\renewcommand{\RCertOpTol}{10}
\renewcommand{\RCertStressShapes}{6}
\renewcommand{\RCertStressRho}{0.14}
\renewcommand{\RCertStressDipErr}{47}
\renewcommand{\RCertStressCorrErr}{43}
\renewcommand{\RGtCases}{4}
\renewcommand{\RGtErrOurs}{45}
\renewcommand{\RGtErrDip}{72}
\renewcommand{\RGtErrCorr}{65}
\renewcommand{\RGtMeanOurs}{53}
\renewcommand{\RGtMeanDip}{56}
\renewcommand{\RGtMeanCorr}{53}
\renewcommand{\RGtToneOurs}{0.143}
\renewcommand{\RGtToneDip}{0.180}
\renewcommand{\RGtToneCorr}{0.168}
\renewcommand{\RGtToneHybrid}{0.132}

\renewcommand{\RGtWinOurs}{1}
\renewcommand{\RGtWinDip}{1}
\renewcommand{\RGtWinCorr}{2}

\renewcommand{\RGtReduce}{38}
\renewcommand{\RGtThickX}{1.6}
\renewcommand{\RGtSpp}{6}
\renewcommand{\RGtBounce}{64}
\renewcommand{\RGtPtSec}{2.2}
\renewcommand{\RGtOursSec}{0.8}

\renewcommand{\RGswShapes}{6}
\renewcommand{\RGswFitTau}{0.38}
\renewcommand{\RGswFitCurv}{0.94}
\renewcommand{\RGswCoefCurv}{+1.3}
\renewcommand{\RHstNoiseMin}{5}
\renewcommand{\RHstNoiseMax}{26}
\renewcommand{\RHstRatioMax}{32}
\renewcommand{\RHstExpG}{0.6}
\renewcommand{\RHstSdfExp}{2.1}
\renewcommand{\RHetErr}{0.35}
\renewcommand{\RHetCRatio}{2}
\renewcommand{\RShadowDrop}{47}
\renewcommand{\RHeroRes}{720}
\renewcommand{\RHeroSec}{0.63}
\renewcommand{\RHeroWalks}{512}
\renewcommand{\RNMedia}{4}
}{}

\newcommand{\sa}{\sigma_a}
\newcommand{\sgs}{\sigma_s}
\newcommand{\stp}{\sigma_t'}
\newcommand{\ssp}{\sigma_s'}
\newcommand{\ld}{\ell_d}
\newcommand{\Tr}{\mathrm{Tr}}
\newcommand{\E}{\mathbb{E}}
\newcolumntype{L}[1]{>{\raggedright\arraybackslash}p{#1}}

\title{\vspace{-1.5em}Dipole Diffusion Error in Thin Geometry:\\
Optical Thickness Laws for Grid-Free Subsurface Scattering}
\author{Faruk Alpay\thanks{Corresponding author: \texttt{alpay@lightcap.ai}} \quad Barış Başaran\\[3pt]
  \small Department of Computer Engineering, Bahçeşehir University, Istanbul, Turkey\\[-1pt]
  \small \texttt{\{faruk.alpay, baris.basaran\}@bahcesehir.edu.tr}}
\date{}

\begin{document}
\maketitle

\begin{abstract}
\noindent
The dipole and its descendants model subsurface scattering with a radial
reflectance profile fitted to a flat, semi-infinite slab. This assumption
introduces a systematic geometry error on thin and curved objects. We isolate the
effect by comparing the dipole with the finite-slab multipole under the same
diffusion model and boundary condition. In slab geometry the diffuse-albedo error
has a material-independent leading rate, $C e^{-2\tau}$ with $\tau=T/\ld$, while
the prefactor remains material dependent; the same image series gives the
transmitted flux, whose leading decay is $e^{-\tau}$. We give the closed-form
albedo and transmittance, relate the exponents to killed random walks, and extend
the interpretation to spatially varying media through optical distance. A
brute-force volumetric path tracer fits a reflectance-deficit rate of
$\RReflRateBrute$ and a transmittance rate of $\RTrRateBrute$, matching the
round-trip and single-pass predictions. The resulting thickness predictor is a
useful thin-feature heuristic, but stress tests show that curvature and
illumination can dominate away from the slab setting. For the remaining
geometry-dependent term we solve the screened-Poisson diffusion problem directly
inside the signed-distance domain with Walk on Spheres, without an interior mesh
or a tangent half-space approximation; the estimator matches closed-form tests to
$\RMaxRelErr\%$. Against a four-case path-traced benchmark it improves the
back-lit, thickness-governed case but not every front-lit or curved case, showing that the method reduces geometry error within diffusion and does not
replace radiative transport.
\end{abstract}

\section{Introduction}
The appearance of a translucent object is produced by light that enters the
surface, scatters repeatedly within the medium, and exits at a different point. Simulating this
transport with brute-force Monte Carlo is correct but slow, so real-time and
production rendering instead use the \emph{diffusion approximation}: in a
strongly scattering medium the multiply-scattered field is well described by a
diffusion equation, and Jensen et al.~\cite{jensen2001} solved that equation in
closed form for a point source on a semi-infinite slab. The result, the
\emph{dipole}, gives a radial diffuse-reflectance profile $R_d(r)$ that has
become a standard component of subsurface rendering. Its refinements, the
multipole~\cite{donner2005}, quantized diffusion~\cite{deon2011}, photon-beam
diffusion~\cite{habel2013}, the directional dipole~\cite{frisvad2014}, the
normalized-diffusion profiles used in film~\cite{christensen2015,burley2015}, and
the separable kernels that put these in screen space for real-time
engines~\cite{jimenez2015separable}, all share the same underlying object. Each is a profile that
depends only on the surface distance $r$ between the points where light enters and
leaves.

\begin{figure}[!tbp]
\centering
\includegraphics[width=0.6\linewidth]{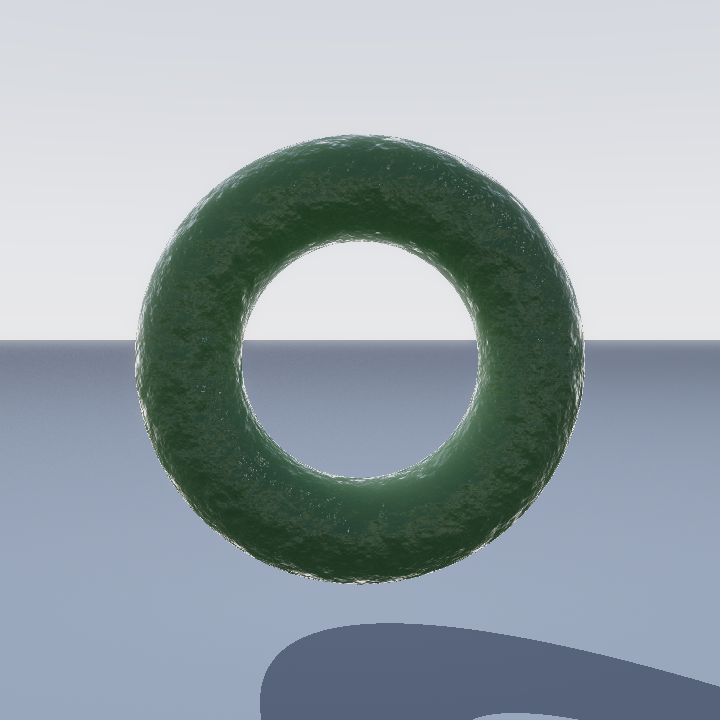}
\caption{A back-lit jade ring rendered with the grid-free diffusion estimator of
Section~\ref{sec:method} beneath a microfacet dielectric coat, under environment
lighting. The thin tube transmits light without any half-space assumption, because the
walk reaches the true second boundary through the distance field. Rendered at
$\RHeroRes^2$ with $\RHeroWalks$ walks per pixel in $\RHeroSec$\,s on an \RDevice.}
\label{fig:hero}
\end{figure}

A profile in $r$ alone encodes a strong geometric assumption. It is exact for a
flat half-space, an infinite reservoir of material below the surface. Real
assets violate it wherever they are thin or sharply curved. At the rim of a leaf
or the edge of a gem there is no half-space below the surface; there is a second
boundary a fraction of a diffusion length away, through which light escapes
instead of returning. The graphics community knows the dipole ``fails on thin
geometry'', and the multipole~\cite{donner2005} was introduced
to handle the planar thin-slab case; more recently a learned model
corrects the profile per shape from training data~\cite{vicini2019}. But the
error has not been characterized as a scalar function of geometry, and a general
geometry-aware diffusion estimator normally requires an interior mesh or grid.
Grid-free Monte Carlo solvers for elliptic boundary-value problems
~\cite{sawhney2020,sawhney2022,sawhney2023,miller2024,sawhney2025sota} remove
that discretization requirement, making it possible to evaluate the diffusion
operator directly on signed-distance geometry.

Our goal is to make this failure mode measurable and to connect that measurement
to a solver used only where a profile is insufficient. Real-time and production
renderers gather a profile in screen or texture space
~\cite{jimenez2015separable,burley2015} and often maintain a local thickness,
whether as a baked map or a screen-space estimate. We show that the same
thickness, measured as $\tau=T/\ld$, determines the leading exponential rate of
the profile error in the planar case. The closed-form ratio
$\kappa(\tau)=R(\tau\ld)/R_\infty$ gives a
per-pixel albedo correction for front-lit over-glow, while the companion
transmittance $\Tr(\tau)$ gives the back-lit flux. Spatially varying media,
curvature-dependent residuals and inverse coefficient recovery are then handled
by the grid-free diffusion estimator.

Four contributions follow.
\begin{itemize}
\item \textbf{A slab law, a predictor, and a thickness correction.} For a finite
slab, the dipole's albedo error has a material-independent exponential rate in
the optical thickness $\tau=T/\ld$, which we derive in closed form:
$C e^{-2\tau}$ with a material-dependent prefactor (Section~\ref{sec:disc},
Proposition~\ref{prop:law}). A brute-force path tracer directly fits the
reflectance-deficit rate as $\RReflRateBrute$, close to the predicted round-trip
rate $2$. The error is the light a thin feature loses to a second surface and the
dipole retains; this interpretation extends the rate from a slab to layered,
spatially varying media, where optical distance replaces straight-line thickness
(Section~\ref{sec:optical}). A thickness map supplies a per-pixel albedo
correction and a useful dispatch heuristic, but Section~\ref{sec:cert} and
Section~\ref{sec:disc2} make explicit where curvature and illumination limit it.
\item \textbf{A geometry-exact estimator.} We evaluate the diffusion equation
inside the true object, given only as a signed-distance function, by Walk on
Spheres (Section~\ref{sec:method}), which requires neither an interior mesh nor a half-space assumption and reproduces
the closed form to $\RMaxRelErr\%$. Against an independent path-traced
reference it moves substantially closer than the dipole, cutting the image error by
$\RGtReduce\%$ overall and $\RGtThickX\times$ in thick regions; the residual is
diffusion-model error rather than solver error (Section~\ref{sec:groundtruth}). The
implementation is portable, plain-float Metal with a CPU reference.
\item \textbf{A solver for the interior field.} Because the unknown is the field
inside the object, the estimator represents effects outside a radial profile. These include
spatially varying media (Section~\ref{sec:hetero}), the near-field
single scattering the diffusion omits, and a dielectric, environment-lit
appearance model that meets the distant-lighting setting of recent
neural~\cite{neupress2024} and path-traced~\cite{werner2024} work without
training or temporal reservoir sampling.
\item \textbf{Differentiable and therefore optimizable.} The forward map is strongly
smoothing, so general structure recovery is ill-posed; but because the solver is
differentiable, a low-dimensional hidden inclusion can be fit from diffuse
measurements, a grid-free form of diffuse optical tomography
(Section~\ref{sec:inverse}), located to $\RInvCentreErr\%$ of the object's radius,
with its depth-limited resolution made explicit.
\end{itemize}

The common representation is probabilistic. In a strongly scattering object the
diffuse field can be written as an expectation over killed random walks that
terminate by absorption or by reaching the boundary~\cite{muller1956}. The dipole
uses the walk law of a half-space; a finite or curved object changes the boundary
hitting probabilities. Evaluating those walks on the signed-distance domain gives
the forward renderer, and differentiating the same map gives the low-dimensional
inverse problem considered in Section~\ref{sec:inverse}.

\section{Background and notation}
\label{sec:bg}
We summarize the classical diffusion model to fix notation; a full derivation is
in Ishimaru~\cite{ishimaru1978} and Stam~\cite{stam1995}.

\paragraph{Diffusion approximation.}
A homogeneous medium has absorption $\sa$, scattering $\sgs$ and
Henyey--Greenstein anisotropy $g$. The similarity relation gives the reduced
scattering $\ssp=\sgs(1-g)$ and the reduced extinction $\stp=\sa+\ssp$. In the
diffusion regime ($\ssp\gg\sa$) the fluence $\phi$ obeys
\begin{equation}
  D\,\nabla^2\phi(\mathbf{x}) - \sa\,\phi(\mathbf{x}) + Q(\mathbf{x}) = 0,
  \qquad D=\frac{1}{3\stp},
  \label{eq:diffusion}
\end{equation}
where $Q$ is the volumetric source of first-scattered light. Dividing by $D$
casts \eqref{eq:diffusion} as a screened Poisson (modified Helmholtz) equation
\begin{equation}
  \nabla^2\phi - c\,\phi = -f,\qquad c=\frac{\sa}{D}=3\sa\stp,\quad f=\frac{Q}{D},
  \label{eq:screened}
\end{equation}
whose screening constant defines the \emph{diffusion length}
$\ld = 1/\sqrt{c} = 1/\sqrt{3\sa\stp}$, the natural length over which diffuse
light persists, and the unit in which we will measure thickness.

\paragraph{Boundary condition.}
At the surface, diffusion theory uses the extrapolated-boundary condition: the
fluence is taken to vanish on a virtual boundary offset outward by
$z_b = 2AD$, where $A=(1+F_{dr})/(1-F_{dr})$ and $F_{dr}$ is the diffuse Fresnel
reflectance for relative index $\eta$~\cite{jensen2001}. The exitant radiant
exitance is then $M=\phi/(2A)$ and, for the Lambertian diffuse term,
$L_o = (1-F_{dr})\,\phi/(2\pi A)$. We adopt exactly this boundary condition for
\emph{every} method we compare, so that differences come only from the geometry
of the bulk and not from the boundary treatment.

\paragraph{Dipole and multipole.}
For a point source at depth $z_r=1/\stp$ on a semi-infinite slab, satisfying the
extrapolated boundary with one mirror source gives the dipole profile
\begin{equation}
  R_d(r)=\frac{\alpha'}{4\pi}\!\left[
  z_r(\kappa d_r{+}1)\frac{e^{-\kappa d_r}}{d_r^3}
  + z_v(\kappa d_v{+}1)\frac{e^{-\kappa d_v}}{d_v^3}\right],
  \label{eq:dipole}
\end{equation}
with $\alpha'=\ssp/\stp$, $\kappa=1/\ld$, $z_v=z_r+2z_b$, and $d_{r,v}=\sqrt{r^2+z_{r,v}^2}$.
A slab of finite thickness $T$ has two boundaries; satisfying both requires an
infinite train of image dipoles, the multipole~\cite{donner2005},
\begin{equation}
  R_d(r;T)=\frac{\alpha'}{4\pi}\sum_{i=-\infty}^{\infty}\!\Big[
  z_{r,i}(\kappa d_{r,i}{+}1)\frac{e^{-\kappa d_{r,i}}}{d_{r,i}^3}
  - z_{v,i}(\kappa d_{v,i}{+}1)\frac{e^{-\kappa d_{v,i}}}{d_{v,i}^3}\Big],
  \label{eq:multipole}
\end{equation}
with $z_{r,i}=2i(T{+}2z_b)+z_r$ and $z_{v,i}=2i(T{+}2z_b)-z_r-2z_b$. As
$T\!\to\!\infty$ the $i{=}0$ term recovers \eqref{eq:dipole}; the remaining terms
are exactly the flux the dipole keeps but a slab of thickness $T$ loses through
its back face. Equation~\eqref{eq:multipole} is the \emph{geometry-aware}
diffusion answer on the one geometry that admits a closed form, and the estimator
of Section~\ref{sec:method} reproduces it.

\section{The slab error rate is governed by optical thickness}
\label{sec:disc}

\begin{figure}[!tbp]
\centering
\includegraphics[width=\linewidth]{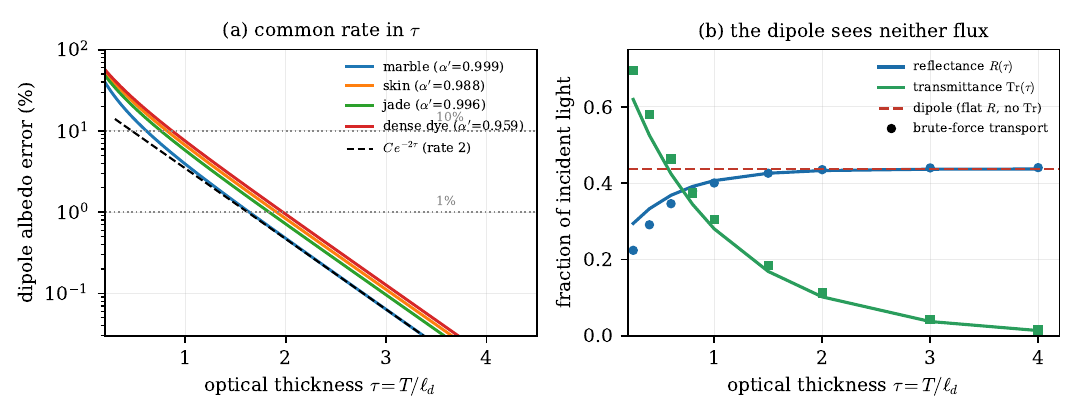}
\caption{\textbf{(a)} The dipole's diffuse-albedo error, measured against the
finite-thickness diffusion solution \eqref{eq:multipole} under the same boundary
condition, on a log scale against $\tau=T/\ld$ for four media. The curves are
near-parallel lines of slope $-2$, tracking the closed-form law $C\,e^{-2\tau}$ of
Proposition~\ref{prop:law} (dashed). Here $\tau$ sets the rate, while the
prefactor varies with material. \textbf{(b)} The slab's two fluxes as fractions of
incidence: the reflectance $R(\tau)$ (Proposition~\ref{prop:law}) and the
companion transmittance $\mathrm{Tr}(\tau)$ (Proposition~\ref{prop:trans}), each
closed form (lines) confirmed by brute-force photon transport (markers). As the
slab thins, reflectance drops and transmittance rises at the single-pass rate;
the dipole predicts a constant reflectance (dashed) and no transmission, so it
misses both.}
\label{fig:discovery}
\end{figure}

\paragraph{Isolating the geometry.}
The dipole \eqref{eq:dipole} and the multipole \eqref{eq:multipole} are the same
diffusion model with the same diffusion constant and the same extrapolated
boundary; they differ only in whether a back face exists. Their difference therefore measures the geometry error alone, since no Fresnel convention, phase function or single-scattering term enters. We quantify
it with the diffuse albedo $R=\int_0^\infty R_d(r)\,2\pi r\,dr$, the fraction of
incident diffuse power re-emitted, and report the relative error
$|R^{\mathrm{dip}}-R^{\mathrm{multi}}(T)|/R^{\mathrm{multi}}(T)$.

\paragraph{A single rate variable.}
Figure~\ref{fig:discovery}a sweeps the thickness for marble, skin, jade and a
dense dye, materials whose reduced albedos span $0.96$ to $0.999$ and whose
diffusion lengths differ by more than an order of magnitude
(Table~\ref{tab:media}). Plotted against $\tau=T/\ld$, the error curves become
near-parallel and share the same asymptotic slope, but they do not have identical
absolute magnitude: at fixed $\tau$ the spread is the material prefactor in
Proposition~\ref{prop:law}. For marble, the dipole over-predicts the albedo by
$\RErrQuarter\%$ at $\tau=0.25$, by $\RErrHalf\%$ at $\tau=0.5$ and by
$\RErrOne\%$ at $\tau=1$; across all four media it falls below $10\%$ only for
$\tau\gtrsim\RTauTenLo$--$\RTauTenHi$ and below $1\%$ only for
$\tau\gtrsim\RTauOneLo$--$\RTauOneHi$. Thus optical thickness determines the
exponential rate and the useful crossover scale, while material controls a
roughly twofold prefactor.

\begin{table}[!tbp]
\centering
\caption{Media used in Figure~\ref{fig:discovery}a and the thickness at which the
dipole's diffuse-albedo error crosses $10\%$ and $1\%$. The crossover thicknesses
differ because $\tau=T/\ld$ sets the rate while the material
prefactor changes the absolute error.}
\label{tab:media}
\small
\begin{tabular}{lcccccc}
\toprule
Medium & $\alpha'$ & $\ell_d$ (mm) & err@$\tau{=}0.25$ & err@$\tau{=}0.5$ & $\tau_{10\%}$ & $\tau_{1\%}$ \\
\midrule
Marble & 0.9992 & 7.78 & 32\% & 14\% & 0.68 & 1.84 \\
Skin & 0.9882 & 4.82 & 45\% & 22\% & 0.87 & 2.09 \\
Jade & 0.9959 & 7.44 & 40\% & 19\% & 0.77 & 1.84 \\
Dense dye & 0.9589 & 0.39 & 48\% & 24\% & 0.98 & 2.09 \\
\bottomrule
\end{tabular}

\end{table}

The curves share the same slope because the deficit follows from the
surface integral of the screened Green's function. We use the following identity
throughout the derivation.

\begin{lemma}
\label{lem:flux}
For a signed image-source depth $z\ne0$ and $d(r,z)=\sqrt{r^2+z^2}$,
\begin{equation}
\int_0^\infty z\,(\kappa d+1)\frac{e^{-\kappa d}}{d^3}\,2\pi r\,dr
=2\pi\,\mathrm{sgn}(z)e^{-\kappa |z|}.
\label{eq:fluxidentity}
\end{equation}
\end{lemma}

\begin{proof}
With $d=\sqrt{r^2+z^2}$, $r\,dr=d\,dd$ and
\[
(\kappa d+1)\frac{e^{-\kappa d}}{d^2}
=-\frac{d}{dd}\left(\frac{e^{-\kappa d}}{d}\right).
\]
The integral from $d=|z|$ to $\infty$ is therefore $e^{-\kappa |z|}/|z|$,
and multiplication by $2\pi z$ gives \eqref{eq:fluxidentity}.
\end{proof}

\begin{proposition}
\label{prop:law}
The geometry-aware diffuse albedo of a slab of thickness $T$ is
\begin{equation}
R(T)=\frac{\alpha'}{2}\,
\frac{(a-q/a)+(b-q/b)}{1-q},\quad
a=e^{-\kappa z_r},\; b=e^{-\kappa(z_r+2z_b)},\; q=e^{-2\kappa(T+2z_b)},
\label{eq:Rslab}
\end{equation}
with the dipole as its thick limit $R_\infty=R(\infty)=\tfrac{\alpha'}{2}(a+b)$. The
deficit of the finite slab relative to that limit is therefore
\begin{equation}
1-\frac{R(T)}{R_\infty}=C\,e^{-2\tau}+\mathcal{O}(e^{-4\tau}),\qquad
C=\frac{(1/a+1/b)-(a+b)}{a+b}\,e^{-4\kappa z_b},
\label{eq:errlaw}
\end{equation}
where $\tau=\kappa T=T/\ld$.
\end{proposition}

\begin{proof}
Let $L=T+2z_b$. Applying Lemma~\ref{lem:flux} to the real image depths
$z_{r,i}=2iL+z_r$ gives
\[
\sum_i \mathrm{sgn}(z_{r,i})e^{-\kappa |z_{r,i}|}
=\sum_{i=0}^{\infty}a q^i-\sum_{i=1}^{\infty}\frac{q^i}{a}
=\frac{a-q/a}{1-q}.
\]
The virtual images enter \eqref{eq:multipole} with the opposite sign. Since
$z_{v,i}=2iL-z_r-2z_b$, their contribution is
\[
-\sum_i \mathrm{sgn}(z_{v,i})e^{-\kappa |z_{v,i}|}
=\sum_{i=0}^{\infty}b q^i-\sum_{i=1}^{\infty}\frac{q^i}{b}
=\frac{b-q/b}{1-q}.
\]
Multiplying the sum by $\alpha'/2$ yields \eqref{eq:Rslab}. Subtracting from
$R_\infty$ gives
\[
1-\frac{R(T)}{R_\infty}
=\frac{q}{1-q}\,
\frac{(1/a+1/b)-(a+b)}{a+b},
\]
and using $q=e^{-2\tau}e^{-4\kappa z_b}$ gives the expansion
\eqref{eq:errlaw}.
\end{proof}

The relative error reported in Section~\ref{sec:disc} normalizes by the truth
$R(T)$ rather than by $R_\infty$; it equals $R_\infty/R(T)-1$, which shares the
leading term $C\,e^{-2\tau}$ and the rate but differs from \eqref{eq:errlaw} at
$\mathcal{O}(e^{-4\tau})$, a difference that is appreciable only in the thinnest
regime ($\tau\lesssim0.5$). The asymptotic rate is the same under either
normalization; the precise percentages are those of the reported quantity. The
error decays at the rate $2$ for every medium; the material enters only
through the prefactor $C$, and through it only via $\kappa z_r=\sqrt{3(1-\alpha')}$
and $\kappa z_b$. The thickness at which the error reaches a level $\varepsilon$ is
thus $\tau_\varepsilon\approx\tfrac12\ln(C/\varepsilon)$, logarithmic in $C$:
across the four media here $C$ spans only $0.26$ to $0.51$, so the one-percent
crossover moves over the narrow range $\tau\approx1.6$ to $2.0$. The
near-parallel family in Figure~\ref{fig:discovery}a reflects a
universal rate together with a logarithmically weak material
prefactor, and \eqref{eq:Rslab} supplies in closed form the correction
$\kappa(\tau)=R(\tau\ld)/R_\infty$ used in Section~\ref{sec:cert}.

\paragraph{Brute-force transport confirms the law.}
The closed-form comparison above still compares two diffusion solutions.
Figure~\ref{fig:discovery}b therefore adds a brute-force volumetric
path tracer (Section~\ref{sec:method}) that uses true Fresnel boundaries and the
Henyey--Greenstein phase function and makes no diffusion approximation. In the
thick limit its reflectance, $\RRbruteThick$, matches the analytic dipole,
$\RRdipole$, to $\RThickAgree\%$, an independent check that both the model and
our conventions are sound. As the slab thins, real transport sheds light through
the back face exactly as the geometry-aware diffusion predicts: at
$\tau=\RTauThin$ it re-emits $\RBruteThinDrop\%$ less than the thick limit, so the
dipole, which predicts no change, over-states the reflectance by
$\RDipOverThin\%$. Fitting the brute-force reflectance deficit
$1-R_{\mathrm{brute}}(\tau)/R_{\mathrm{brute}}(\infty)$ over
$0.6\le\tau\le3$ gives rate $\RReflRateBrute$, against $\RReflRateClosed$ for the
closed form and the predicted round-trip rate $2$; the closed-form deficit
matches the transport deficit to $\RReflAgree\%$ in that band. Because it compares fitted decay rates rather than flux magnitudes, this check tests the exponent of the headline reflectance law directly. Figure~\ref{fig:profiles} decomposes the brute-force profile
by the number of scattering events. The diffusion model
describes \emph{multiple} scattering, and it matches the brute-force
multiple-scattering component (markers) across the profile. It omits \emph{single} scattering, light that scatters exactly once before leaving, a sharp near-field term. This single-scattering term is $\RSingleTotal\%$ of all re-emitted light overall, $\RSingleNear\%$ of it within one diffusion length of the entry point, and essentially absent ($\RSingleFar\%$) beyond three. This near-field peak is absent from a diffusion-only profile, and we add it back in the renderer (Section~\ref{sec:single}).

\begin{figure}[!tbp]
\centering
\includegraphics[width=\linewidth]{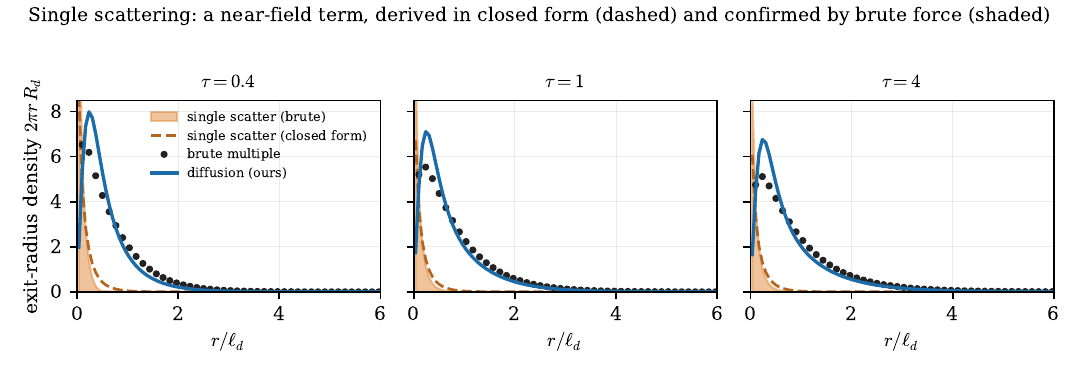}
\caption{Brute-force exit-radius density $2\pi r R_d(r)$ decomposed by scattering
order, for three thicknesses. The diffusion estimator (ours, line) matches the
brute-force \emph{multiple}-scattering component (markers) in the mid- and
far-field; the \emph{single}-scattering term is a sharp near-field peak the
diffusion misses ($\RSingleNear\%$ of re-emission within one diffusion length,
$\RSingleFar\%$ beyond three), which the closed form \eqref{eq:single} (dashed)
reproduces over the brute-force measurement (shaded). The renderer adds it by a
ray-march (Section~\ref{sec:single}).}
\label{fig:profiles}
\end{figure}

\paragraph{Transmittance.}
The light a thin slab fails to reflect is transmitted through it, and the same
image series gives the flux through the far face in closed form. This term
accounts for the back-lit appearance of thin features (Figure~\ref{fig:hero}).

\begin{proposition}
\label{prop:trans}
The diffuse transmittance of a slab of thickness $T$ is
\begin{equation}
\Tr(T)=\frac{\alpha'}{2}\,e^{-\kappa T}\,
\frac{(1/a-b)+p\,(1/b-a)}{1-q},\qquad p=e^{-4\kappa z_b},
\label{eq:trans}
\end{equation}
with $a,b,q$ as in Proposition~\ref{prop:law}. Its leading behaviour is
$\Tr\sim e^{-\tau}$.
\end{proposition}

\begin{proof}
The flux through the far face is obtained by applying
Lemma~\ref{lem:flux} to the same image train, but with each signed depth measured
from the plane $z=T$ rather than from the entrance plane. The real images give
$e^{-\kappa T}(1/a-b)/(1-q)$, while the virtual images give the reflected train
$e^{-\kappa T}p(1/b-a)/(1-q)$, with $p=e^{-4\kappa z_b}$. Multiplying by the
same $\alpha'/2$ prefactor yields \eqref{eq:trans}. The numerator and denominator
are independent of $T$ except through $q=\mathcal{O}(e^{-2\tau})$, so the leading
term is proportional to $e^{-\tau}$.
\end{proof}

The two rates admit a direct interpretation. Transmittance corresponds to a
single pass through the slab and decays as $e^{-\tau}$, whereas the reflectance
error of Proposition~\ref{prop:law} corresponds to a round trip and decays as
$e^{-2\tau}$. Both quantities follow from the same image series, so a single
diffusion operator accounts for the dipole's error and for the transmission that
produces back-lit appearance. The brute-force photon tracer confirms both closed
forms directly (Figure~\ref{fig:discovery}b): its bottom-face flux tracks
$\Tr(\tau)$ to $\RTrAgree\%$ for $\tau\gtrsim0.5$, with a fitted decay
rate of $\RTrRateBrute$, the single-pass rate $1$ the derivation predicts, while
the dipole transmits nothing at all. Below $\tau\!\approx\!0.5$ the slab transmits
ballistically, faster than diffusion predicts, the same regime in which the
diffusion approximation itself loosens~\cite{sawhney2025sota}.

\subsection{What the over-glow is, and where it generalizes}
\label{sec:optical}

The surface glow is the light gathered from everywhere inside the object. The
dipole gathers it as though the object were a flat half-space, so its error equals
the light the real object loses through a nearby second surface, which is absent from the
half-space model,
\begin{equation}
e(\mathbf{x})=\int_\Omega\big[G_\Omega(\mathbf{x},\mathbf{y})-G_H(\mathbf{x},\mathbf{y})\big]
\,q(\mathbf{y})\,d\mathbf{y},
\label{eq:greendiff}
\end{equation}
the difference between the response of the true object $G_\Omega$ and of a
half-space $G_H$. The slab multipole is the special case where this difference is
the train of mirror images across the far face.

The random-walk representation quantifies that loss. For the
operator $\nabla^2-c$ in \eqref{eq:screened}, the Feynman--Kac formula writes the
survival contribution to a boundary set $\Gamma$ as
\begin{equation}
u(\mathbf{x})=
\E_{\mathbf{x}}\!\left[
\exp\!\left(-\int_0^{\tau_\Gamma} c(B_s)\,ds\right)
\mathbf{1}_{\tau_\Gamma<\tau_{\partial\Omega\setminus\Gamma}}
\right],
\label{eq:fk}
\end{equation}
up to the conventional time scaling of Brownian motion~\cite{oksendal2003}. The
half-space dipole retains paths that, in the real object, hit a far boundary and
leave. In a homogeneous slab, let $u(z)$ be the probability weight that such a
walk, started at depth $z$ with the front surface reflecting, reaches the far face
at depth $T$ before absorption. It satisfies $u''-c\,u=0$, $u'(0)=0$, $u(T)=1$,
so
\begin{equation}
u(0)=\operatorname{sech}\tau,\qquad \tau=\sqrt{c}\,T=T/\ld .
\label{eq:sech}
\end{equation}
One crossing fades as $e^{-\tau}$, the back-lit transmittance of
Proposition~\ref{prop:trans}; the round trip $u(0)^2\sim e^{-2\tau}$ is the
reflectance error of Proposition~\ref{prop:law}. The random walk and the image series yield the same two rates, and the closed form matches the simulated walk to $\ROptSech$ (Figure~\ref{fig:optical}a).

This view extends to media the image series cannot treat. When the material varies inside the object there is no single diffusion length, but the exponential part of \eqref{eq:fk} is governed by the least costly path to the far boundary. The
corresponding optical distance is
\begin{equation}
\tau_\star(\mathbf{x})=\inf_{\gamma:\,\mathbf{x}\to\Gamma_{\mathrm{far}}}
\int_\gamma\sqrt{c(\mathbf{s})}\,\,ds ,
\label{eq:agmon}
\end{equation}
the Agmon distance of the screened operator~\cite{agmon1982}. At leading
exponential order, $u(\mathbf{x})$ scales like $\exp[-\tau_\star(\mathbf{x})]$,
so the reflected loss scales like $\exp[-2\tau_\star(\mathbf{x})]$. Writing the survival function as $u=Ae^{-S}$ in
$\nabla^2u-c(\mathbf{x})u=0$ gives the leading eikonal equation
\begin{equation}
|\nabla S(\mathbf{x})|^2=c(\mathbf{x}),\qquad S|_{\Gamma_{\mathrm{far}}}=0,
\label{eq:eikonal}
\end{equation}
with $S=\tau_\star$ to leading order. The amplitude $A$ is fixed by the next order of the expansion, the transport equation $2\,\nabla S\cdot\nabla A + A\,\nabla^2 S = 0$, and $A$ carries the curvature, source distribution and boundary factors. The optical-distance
collapse holds because $S$ alone sets the exponent while $A$ varies slowly and
folds into the per-material prefactor. This separation requires the medium to
change little over a diffusion length, $\ld\,\lvert\nabla\ln c\rvert\ll1$. Where
$c$ varies on the scale of $\ld$ itself, $\nabla A$ is no longer subdominant and
the leading exponential stops predicting the magnitude. The slab law is the
constant-coefficient, straight-geodesic case of this optical-distance law. We
check it on slabs whose absorption varies smoothly with depth (a deepening
medium, a skin-like falloff, a buried layer): across all of them $-\ln u(0)$
lines up with $\tau_\star$ at correlation $\ROptCorr$ and unit slope, against
$\ROptCorrGeom$ for the straight-line thickness (Figure~\ref{fig:optical}b).
These media vary slowly by construction. High-frequency media that violate
$\ld\,\lvert\nabla\ln c\rvert\ll1$, and a varying diffusion constant that changes
the metric and adds interface terms, lie outside what we test.

\begin{figure}[!tbp]
\centering
\includegraphics[width=\linewidth]{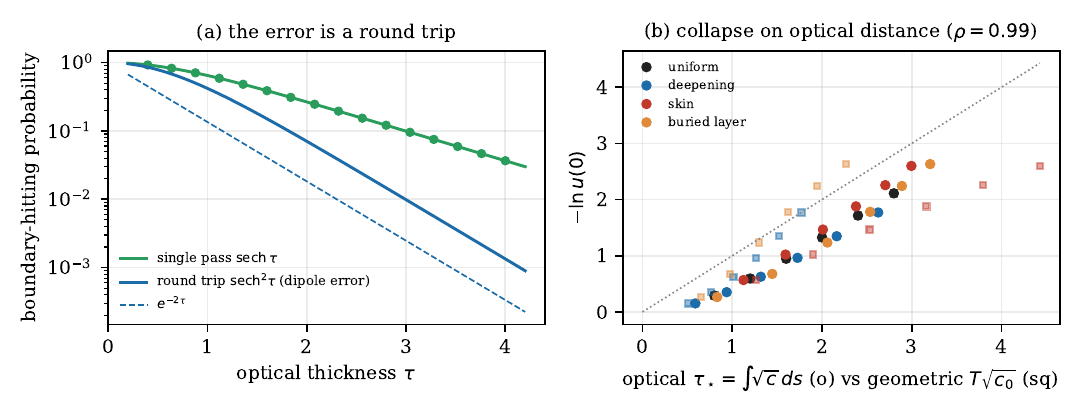}
\caption{The dipole error as a boundary-hitting probability. \textbf{(a)} The
probability that a killed diffusing path reaches the far face, $u(0)$ (markers),
equals $\operatorname{sech}\tau$ \eqref{eq:sech} (green); its square, the round
trip, is the reflectance error and approaches $e^{-2\tau}$ (blue). \textbf{(b)}
For depth-varying screening $c(z)$, $-\ln u(0)$ collapses onto the optical distance
$\tau_\star=\int\!\sqrt{c}\,ds$ \eqref{eq:agmon} (circles, one line) but not onto
the geometric thickness $T\sqrt{c_0}$ (squares, scattered).}
\label{fig:optical}
\end{figure}

\section{Predicting the error before rendering}
\label{sec:cert}
The slab law suggests a practical predictor, but the predictor is not a theorem
for arbitrary geometry. The thickness beneath a surface point is the shape
diameter function~\cite{shapira2008}, the length of a ray cast into the object
opposite the normal; dividing by the diffusion length gives a local
$\tau(\mathbf{x})$, and the closed-form curve gives the expected albedo loss for
a locally slab-like feature. The quantity this predictor truly needs is the
optical distance $\tau_\star$ of Section~\ref{sec:optical}; the shape diameter
over $\ld$ is its homogeneous, nearly flat approximation. In a medium whose
coefficients vary the optical-distance predictor follows the error, whereas the
geometric one does not (Figure~\ref{fig:optical}b). A screened-Poisson walk
estimates this distance, through its survival to the far boundary, as readily as
it estimates the field, so the predictor and the render are computed the same way on the same geometry.

Figure~\ref{fig:cert} compares the predictor against the measured error. We
render a biconvex lens, thick at the centre and thin at the rim, with the dipole
and with the grid-free estimator of the next section, taking the latter as the
geometry-exact reference. The dipole's error is concentrated in a ring at the rim
and is near zero across the body; the predictor reproduces this pattern and
ranks the pixels by error at a Spearman correlation of $\RCertRho$. It predicts
rather than certifies, and it is illumination-agnostic while the true
error is source-weighted and, on strongly curved geometry, curvature-dependent
(Section~\ref{sec:disc2}); we use it to rank where to spend effort, not to
guarantee a tolerance. A stress test on \RCertStressShapes\ analytic SDFs makes
this limitation explicit: the mean rank correlation drops to
$\RCertStressRho$ when curvature or nearly constant thickness dominates, and the
same thickness correction reduces mean error only from $\RCertStressDipErr\%$ to
$\RCertStressCorrErr\%$. Thus the per-pixel ranking works only as a thin-feature heuristic; where
thickness barely varies, the residual is set by curvature and illumination,
which it does not see. The estimator-selection
experiment in Section~\ref{sec:groundtruth} therefore uses only a coarser regime
variable, whether a surface is back-lit, instead of treating the thickness
ranking as a bound. Within the thin-feature regime it was derived for, the
predictor still supports two uses.

The first use needs no additional solver. The closed-form curve that predicts the
error also prescribes its correction. The geometry-aware albedo equals the
dipole's multiplied by $\kappa(\tau)=R(\tau\ld)/R_\infty$
(Proposition~\ref{prop:law}), which decreases from one in the thick limit toward
zero as a feature thins. Scaling the dipole by $\kappa(\tau(\mathbf{x}))$, a
per-pixel thickness lookup, yields a thickness-corrected dipole, and because the
slab albedo ratio depends on $\tau$ a single correction curve applies across
materials up to the prefactor already discussed. On this lens the correction
reduces the mean error from $\RCertDipErr\%$ to
$\RCertCorrErr\%$ and the ninety-fifth percentile from $\RCertDipWC\%$ to
$\RCertCorrWC\%$, without any additional rendering. It corrects the magnitude but not the full
spatial profile, so a residual remains where the geometry is most pronounced.

For that residual the predictor dispatches the grid-free solver. Rendering the
flagged pixels with it and the rest with the corrected dipole, and sweeping the
tolerance, traces the cost--accuracy curve of Figure~\ref{fig:cert}c: grid-free
on the $\RCertOpFrac\%$ of pixels flagged at a $\RCertOpTol\%$ tolerance reaches
$\RCertOpErr\%$, and the curve continues to grid-free accuracy as the budget
grows. The predictor lets one trade cost for accuracy along a single curve, which
a profile-based method does not offer; it ranks where to spend effort, and
because it is a correlation it does not guarantee a tolerance.

\begin{figure}[!tbp]
\centering
\includegraphics[width=\linewidth]{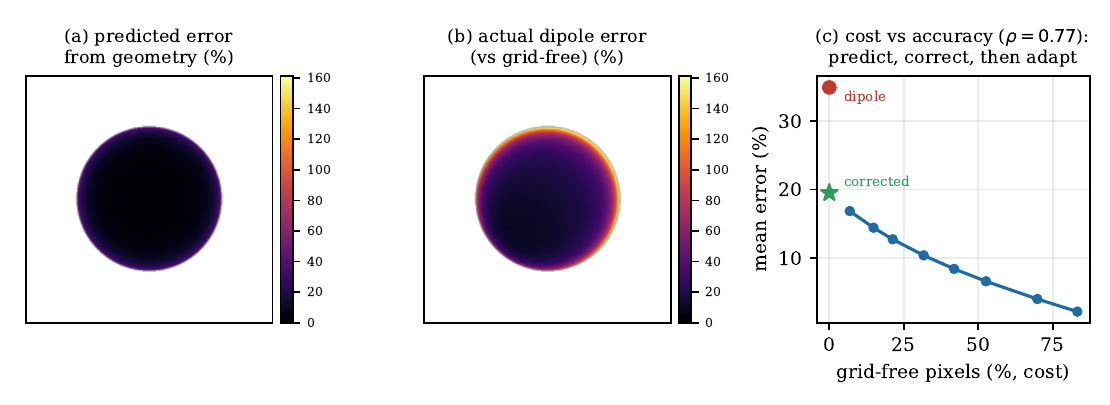}
\caption{Predicting, correcting and dispatching, from the local thickness alone.
\textbf{(a)} the dipole error predicted from the shape diameter function and the
master curve, with nothing rendered; \textbf{(b)} the dipole's actual error
against the grid-free reference, which it matches (Spearman $\RCertRho$).
\textbf{(c)} the same curve gives a thickness correction that halves the
error (green), and the predictor then dispatches the grid-free solver to the
flagged pixels, tracing a cost--accuracy curve down to grid-free accuracy.}
\label{fig:cert}
\end{figure}

\section{A grid-free, geometry-exact estimator}
\label{sec:method}
The finding above is actionable only if the geometry-aware diffusion can be
evaluated on shapes without a closed form. We do this without discretizing the
interior, by solving \eqref{eq:screened} pointwise with Walk on
Spheres~\cite{sawhney2020}.

\paragraph{Walk on Spheres for screened Poisson.}
The solution of \eqref{eq:screened} at a point $\mathbf{x}$ satisfies the
integral identity, over the largest ball $B(\mathbf{x},R)$ that fits inside the
domain,
\begin{equation}
  \phi(\mathbf{x}) = \bar P(R)\,\mathbb{E}_{\partial B}[\phi]
  + \int_{B} G^{c}_R(\mathbf{x},\mathbf{y})\,f(\mathbf{y})\,d\mathbf{y},
  \label{eq:wos}
\end{equation}
where $\bar P(R)=\kappa R/\sinh(\kappa R)$ is the (sub-unit) survival weight of
the modified Helmholtz Poisson kernel, the next point is sampled uniformly on the
sphere $\partial B$, and the ball Green's function is
$G^c_R(t)=\sinh(\kappa(R-t))/(4\pi t\,\sinh(\kappa R))$ with $t=\|\mathbf{y}-\mathbf{x}\|$
(reducing to $(R-t)/(4\pi tR)$ as $c\!\to\!0$). A single estimator of
\eqref{eq:wos} draws one interior point uniformly in the ball for the source
term, multiplies the carried weight by $\bar P(R)$, jumps to the sphere, and
terminates when it reaches an $\varepsilon$-shell of the boundary, where it
collects the (here homogeneous) boundary value. The screening makes $\bar P<1$,
so the walk terminates with probability one even in unbounded directions.

\paragraph{Geometry from a distance field.}
The only thing \eqref{eq:wos} needs from the geometry is the empty-ball radius
$R$, i.e.\ the distance to the boundary. We represent the object by a
signed-distance function $\mathrm{sdf}(\mathbf{x})$ (negative inside) and place
the extrapolated boundary at the level set $\mathrm{sdf}=z_b$, so that
$R(\mathbf{x})=z_b-\mathrm{sdf}(\mathbf{x})$. Because the distance field reports
the nearest surface in \emph{any} direction, the walk automatically respects a
nearby back face or a tight curvature: there is no half-space assumption, and the
estimate is geometry-exact for whatever shape the distance field describes. The
offset boundary $\mathrm{sdf}=z_b$ is a clean shell only where the extrapolation
length stays below the reach of the surface; on features thinner than $2z_b$ the
two offset faces meet near the medial axis, and the diffusion model, like the
dipole, no longer has a well-separated second boundary to resolve. The
source is $f(\mathbf{x})=3\ssp\stp\,E_{\mathrm{ri}}(\mathbf{x})$, the reduced
intensity of the light having entered through the surface and decayed as
$e^{-\stp\ell}$, with $\ell$ the distance back to the entry point found by sphere
tracing along the light direction.

\paragraph{Rendering.}
A camera ray is sphere-traced to the surface; at the hit we evaluate
$\phi$ just inside with $N$ independent walks and convert it to exitant radiance
through $L_o=(1-F_{dr})\,\phi/(2\pi A)$, with one set of coefficients per colour
channel. The same routine produces the dipole's prediction when the distance
field is replaced by the tangent half-space at the hit, $\mathrm{sdf}(\mathbf{x})
=\langle\mathbf{x}-\mathbf{x}_o,\mathbf{n}\rangle$. The estimator, sample count
and source are identical, so the \emph{only} difference between ``ours'' and
``dipole'' in Figure~\ref{fig:compare} is whether the geometry is the true shape or its local flat tangent.

\paragraph{Brute-force reference.}
For ground truth we trace photons analogically: a beam refracts through the
surface (Fresnel), takes exponential free flights, scatters by the
Henyey--Greenstein phase function, is absorbed with probability $1-\sgs/(\sa+\sgs)$
per collision, and reflects or refracts at the boundary with the Fresnel
probability, until it exits or dies. This makes no diffusion approximation and is
the reference used in Section~\ref{sec:disc} and Figure~\ref{fig:discovery}b. By
counting each photon's scattering events the same tracer separates single from
multiple scattering, and produces Figure~\ref{fig:profiles}.

\subsection{Single scattering}
\label{sec:single}
The diffusion term is, by construction, a model of multiple scattering, and
Figure~\ref{fig:profiles} confirms it captures that component. The missing piece
is single scattering, light that refracts in, scatters once and leaves,
which dominates the near field ($\RSingleNear\%$ of re-emission within a diffusion
length) and gives thin, back-lit features their sharp directional appearance. Its reach
is set by the reduced mean free path: a singly scattered photon travels of order
$1/\stp$ before it leaves, which in diffusion lengths is $\sqrt{3(1-\alpha')}$ and
vanishes as the albedo approaches one. Single scattering is therefore a near-field term, confined within a diffusion length and negligible beyond a few. The brute-force single-scatter exits have a mean radius of $0.13\,\ld$, well inside that scale.

Single scattering comes from a different mechanism than the diffusion image series. It is ballistic rather than multiply scattered, and it has its own closed form. A photon
refracting in along the normal, scattering once at depth $z$ and leaving at radial
distance $r$ after an interior path $d=\sqrt{r^2+z^2}$, gives the single-scatter
profile
\begin{equation}
R^{(1)}(r)=\int_0^\infty \sgs\,e^{-\sigma_t z}\,
p\!\left(\tfrac{z}{d}\right)e^{-\sigma_t d}\,\frac{z}{d^{3}}\,dz,
\label{eq:single}
\end{equation}
with $p$ the Henyey--Greenstein phase function and the true (not reduced)
coefficients. Equation~\eqref{eq:single} matches the brute-force single-scatter
profile to a shape correlation of $\RSingleCorr$ (Figure~\ref{fig:profiles},
dashed), and its mass sits at the reduced mean free path, confirming the scale
above. We add it with a short ray-march rather than a walk. At a visible point we
refract the view ray into the medium and step along it; at each step we connect to
the light, attenuating by the transmittance in and out and weighting by the
Henyey--Greenstein phase function and the boundary Fresnel terms, and we add an
isotropic term for the ambient environment. Using the true (not reduced)
coefficients, this single-bounce integral complements the diffusion's multiple
scattering, so the renderer reproduces both the near-field peak and the diffuse
body.

\subsection{Validation, convergence and throughput}
\label{sec:valid}

\begin{figure}[!tbp]
\centering
\includegraphics[width=\linewidth]{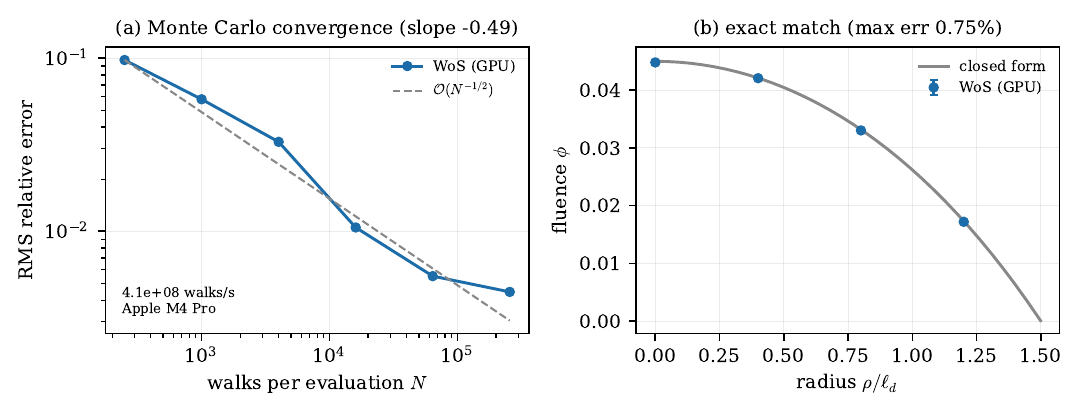}
\caption{\textbf{(a)} RMS relative error of the GPU estimator against the
closed-form ball solution as a function of walks per evaluation; the slope is
$\RConvSlope$, the $\mathcal{O}(N^{-1/2})$ rate of Monte Carlo. The residual
floor at large $N$ is the $\mathcal{O}(\varepsilon)$ bias of the boundary shell.
\textbf{(b)} The estimator (points, $3\sigma$ error bars) lands on the closed-form
fluence inside a screened-Poisson ball, with a maximum error of $\RMaxRelErr\%$.}
\label{fig:validation}
\end{figure}

We verify the estimator against two closed-form problems. For the screened
Poisson equation with a uniform source in a ball, the exact fluence is
$\phi(\rho)=(f_0/c)\,[\,1-(L/\rho)\sinh(\kappa\rho)/\sinh(\kappa L)\,]$;
Figure~\ref{fig:validation}b shows the estimator reproducing it to
$\RMaxRelErr\%$. For a finite slab with an exponential depth source, the
planar geometry where the dipole fails, the three-dimensional walk reproduces
the one-dimensional closed form to under $1\%$. The error falls as
$\mathcal{O}(N^{-1/2})$ (Figure~\ref{fig:validation}a, slope $\RConvSlope$), as
expected for an unbiased Monte Carlo estimator, until it reaches the small
$\mathcal{O}(\varepsilon)$ bias floor of the termination shell. The CPU reference
and the GPU agree to $\RCpuGpuAgree\%$, so the finding does not depend on the GPU.
On an \RDevice\ the plain-float Metal kernel sustains $\RThroughput\times10^{8}$
ball-walks per second.

\subsection{Comparison to a path-traced reference}
\label{sec:groundtruth}

\begin{figure}[!tbp]
\centering
\includegraphics[width=\linewidth]{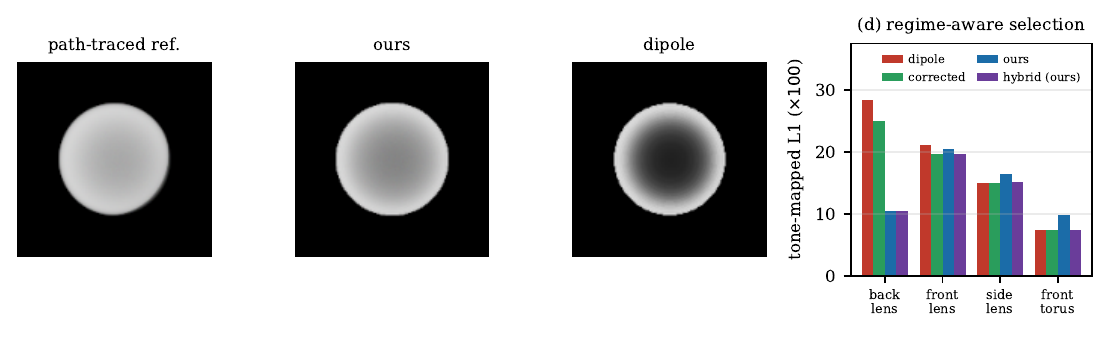}
\caption{Validation against an independent path-traced reference (backward
volumetric path tracing with next-event estimation and angular Fresnel
boundaries, $\RGtSpp$k samples and $\RGtBounce$ bounces per pixel, no diffusion
approximation). The first three panels show the back-lit lens, where appearance
is governed by transmission through local thickness and the half-space profile
leaves the centre too dark. \textbf{(d)} A four-case benchmark (tone-mapped L1)
compares the dipole, the closed-form thickness-corrected dipole, the grid-free
estimator, and a regime-aware hybrid that routes grid-free to back-lit pixels and
the corrected profile elsewhere. No fixed method wins every case; in this
benchmark the hybrid has the lowest mean tone-mapped error.}
\label{fig:groundtruth}
\end{figure}

The comparisons so far measure the dipole against analytic profiles and against
our own diffusion estimator. We close the loop with an independent reference: a
backward volumetric path tracer with next-event estimation and true angular
Fresnel boundaries, which makes no diffusion approximation, validated against the
brute-force tally of Section~\ref{sec:disc}. We render \RGtCases\ cases
(back-lit, front-lit and side-lit lenses, plus a front-lit torus), take the path
tracer ($\RGtSpp$k samples per pixel) as the reference, and compare three
diffusion estimators after exposure matching: the dipole, the same dipole scaled
by $\kappa(\tau)$, and the grid-free solution (Figure~\ref{fig:groundtruth}).

The improvement is narrower than the single-image comparison suggested. On the
back-lit lens, where transmission through thickness is the dominant effect, the
grid-free estimator reduces mean raw error from $\RGtErrDip\%$ to
$\RGtErrOurs\%$; the corrected profile reaches $\RGtErrCorr\%$. Across all
\RGtCases\ cases, however, the mean raw errors are $\RGtMeanDip\%$ (dipole),
$\RGtMeanCorr\%$ (corrected) and $\RGtMeanOurs\%$ (grid-free). On a tone-mapped
RGB L1 metric the ordering is clearer, with means $\RGtToneDip$,
$\RGtToneCorr$ and $\RGtToneOurs$, respectively, but the raw errors remain
large. By raw error, the win counts are \RGtWinOurs/\RGtCases\ for grid-free,
\RGtWinCorr/\RGtCases\ for the corrected profile, and
\RGtWinDip/\RGtCases\ for the dipole.

No single estimator wins everywhere. The corrected profile is accurate and nearly free in cost where a radial profile applies, namely thin, convex, front- or side-lit, homogeneous material. The grid-free solver is worthwhile in the regimes a profile cannot express, of which transmission is the one this benchmark isolates.
Transmission is concentrated where the surface faces away from the light
($\mathbf{n}\cdot\mathbf{l}<0$), and there no profile carries a thickness term at
all while the dipole over-predicts by up to $\RGtThickX\times$ in the thickest
interior. The same structural gap drives the other two settings of this paper,
heterogeneous media (Section~\ref{sec:hetero}) and the inverse problem
(Section~\ref{sec:inverse}), neither of which a single $R_d(r)$ can represent.
Pairing a diffusion model with a costlier solver and choosing between them per
region is standard production practice~\cite{chiang2016,werner2024}; here the
choice is made per pixel from $\mathbf{n}\cdot\mathbf{l}$ with no reference,
routing the grid-free solver to back-lit pixels and the corrected profile
elsewhere. That applies the solver only where no profile suffices and gives the
lowest mean error, tone-mapped L1 $\RGtToneHybrid$ against $\RGtToneOurs$
(grid-free), $\RGtToneCorr$ (corrected) and $\RGtToneDip$ (dipole). The reference
takes $\RGtPtSec$\,s against $\RGtOursSec$\,s for the grid-free image at the same
resolution in the back-lit case.

\section{Results}
\label{sec:results}

The figures are backed by the same result files that generate the in-text
numbers, and no claim rests on a single image. The slab law
uses \RNMedia\ media over forty thicknesses; the transport check uses ten slab
thicknesses and 120 exit-radius bins; solver validation uses a closed-form ball
over six sample counts; the predictor is tested on a lens and on
\RCertStressShapes\ analytic stress shapes; the path-traced comparison covers
\RGtCases\ lighting/geometry cases at $160^2$ pixels and \RGtSpp k reference
samples per pixel; the variable-medium test solves a $64^2$ slice with a
\RHetCRatio$\times$ coefficient ratio; the inverse problem uses \RInvSrc\
sources and \RInvDet\ probes; and the geometry sweep covers \RGswShapes\ shapes
and 89,472 surface samples. The experiments isolate the mechanisms studied in
the paper and are not meant as a benchmark of a production renderer. Every plotted
number, fit, and macro is regenerated from the scripts shipped with the
ancillary files.

\begin{figure}[!tbp]
\centering
\includegraphics[width=\linewidth]{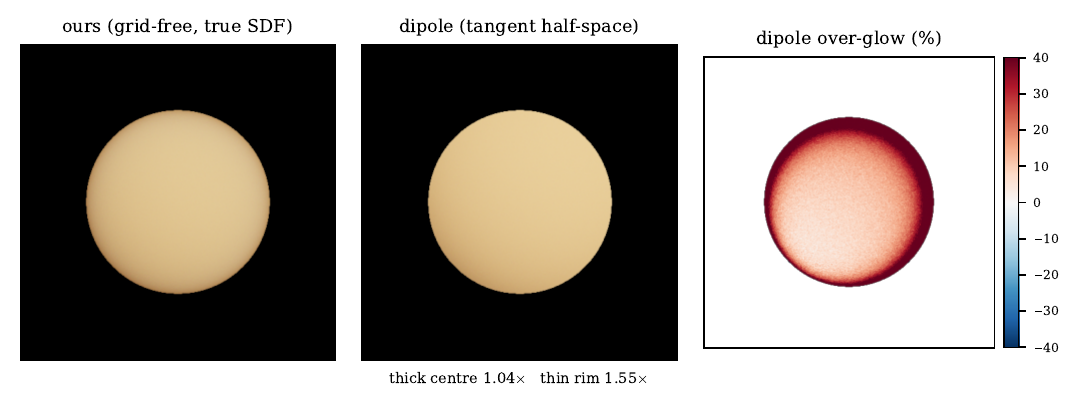}
\caption{A front-lit biconvex lens, thick at the centre, thin at the rim,
rendered with the grid-free estimator (left) and with the tangent-half-space
dipole (middle) at equal sample count and exposure. The signed difference (right)
shows the dipole's error is spatially organised by thickness, consistent with
Section~\ref{sec:disc} predicts: it agrees at the thick centre
($\RCenterRatio\times$) and over-glows the thin rim by $\RRimOverGlowPct\%$
($\RRimOverGlow\times$), because the half-space assumption keeps back-scattered
light that the real thin rim lets escape.}
\label{fig:compare}
\end{figure}

\paragraph{A rendered example.}
Figure~\ref{fig:compare} renders a biconvex lens, which is thick at the centre
and thin at the rim, with the two methods at the same sample count and exposure.
The difference map is the spatially-resolved version of Figure~\ref{fig:discovery}:
near zero at the thick centre and a bright ring of over-glow at the thin rim,
where the dipole is $\RRimOverGlow\times$ too bright. A practitioner shading a
thin-featured asset with the dipole therefore sees a structured brightening of every thin part, which are the parts a
viewer perceives as translucent.

\paragraph{Appearance under environment lighting.}
A translucent object has more than a subsurface term. It also has a dielectric
surface that reflects its surroundings. We wrap the grid-free estimator in a
complete appearance model. The surface is a microfacet
dielectric~\cite{walter2007}. The view ray reflects a glossy, roughness-filtered
sample of a procedural studio environment, weighted by the Schlick--Fresnel
factor, so polished stone shows a sharp highlight and a bright reflective rim.
The complementary transmitted fraction drives the subsurface term, which is now
lit by the same environment rather than a single source. Its distant irradiance
enters the screened-Poisson source alongside the key light, so the unlit side
glows softly with the environment's colour. This matches the distant-lighting,
heterogeneous setting that NeuPreSS~\cite{neupress2024} addresses by training a
neural network and that ReSTIR subsurface scattering~\cite{werner2024} addresses
by path tracing. The grid-free solver produces it directly, with no
precomputation and no training. Figure~\ref{fig:hero} (a back-lit jade ring) and
Figure~\ref{fig:gallery} (marble, jade, amber and coral) are produced by this
model, changing only the distance field and the per-channel coefficients.

\begin{figure}[!tbp]
\centering
\includegraphics[width=0.66\linewidth]{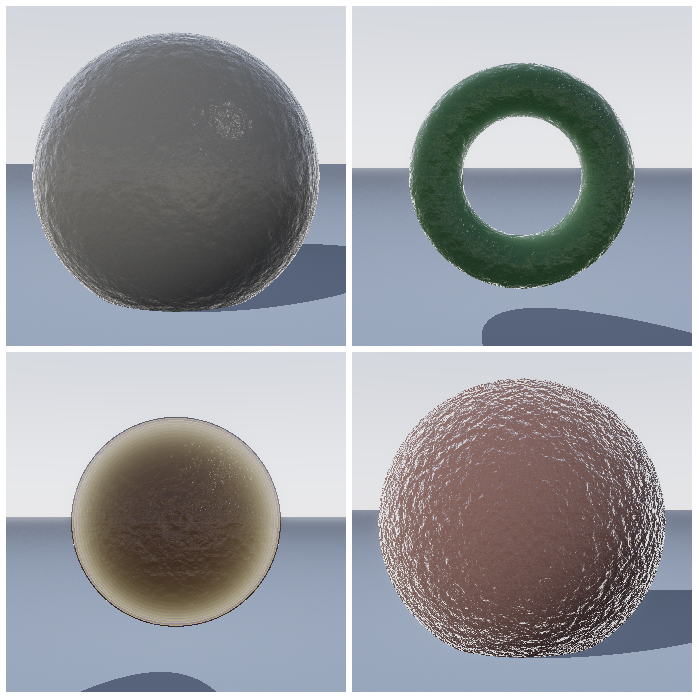}
\caption{Translucent materials under environment lighting, rendered by the
grid-free subsurface estimator beneath a microfacet dielectric coat, changing
only the signed-distance function and the per-channel optical coefficients: a
front-lit marble sphere, a glossy jade ring, a back-lit amber lens and a softer
coral sphere.}
\label{fig:gallery}
\end{figure}

The renders were produced on an \RDevice; none of the reported quantities
(relative error, convergence rate, the $\tau$-dependent slab rate) depend on the hardware,
and the CPU reference reproduces them. The Metal kernels compile offline to a
\texttt{metallib} when a full toolchain is present and at runtime otherwise, so
the same source runs across toolchain versions and Apple GPUs.

\section{Spatially varying media}
\label{sec:hetero}
A reflectance profile $R_d(r)$ describes one homogeneous material and is
therefore restricted to spatially uniform coefficients. The grid-free estimator
instead solves the screened-Poisson problem with the coefficients defined in the
volume. When the screening varies in space the equation becomes
$\nabla^2\phi - c(\mathbf{x})\,\phi = -f(\mathbf{x})$, and we estimate it with a
null-scattering walk in the spirit of Sawhney et al.~\cite{sawhney2022}: the walk
runs with a constant majorant $\bar c \ge c(\mathbf{x})$ and the deficit
$\bar c - c(\mathbf{x})$ is treated as an extra source proportional to $\phi$,
sampled in the same ball. A per-ball majorant, raised only where a ball can reach
the denser region, keeps the deficit, and its variance, local. On a
manufactured solution with $c$ varying $\RHetCRatio\times$ the estimator is exact
to $\RHetErr\%$. For a piecewise inclusion, where the absorption changes but the
scattering (hence the diffusion constant) does not, the fluence is continuous
across the interface and an even simpler interface-limited walk applies.

Figure~\ref{fig:hetero} shows a translucent disk lit from one
side, solved with and without a denser inclusion. The inclusion casts a clear
shadow in the internal light field (the fluence behind it drops by
$\RShadowDrop\%$) and re-shapes the glow around it. A single homogeneous profile
has no coefficient field on which to place that inclusion; the grid-free solver
uses the same estimator as in the homogeneous case, with extra walks concentrated
near the inclusion.

\begin{figure}[!tbp]
\centering
\includegraphics[width=\linewidth]{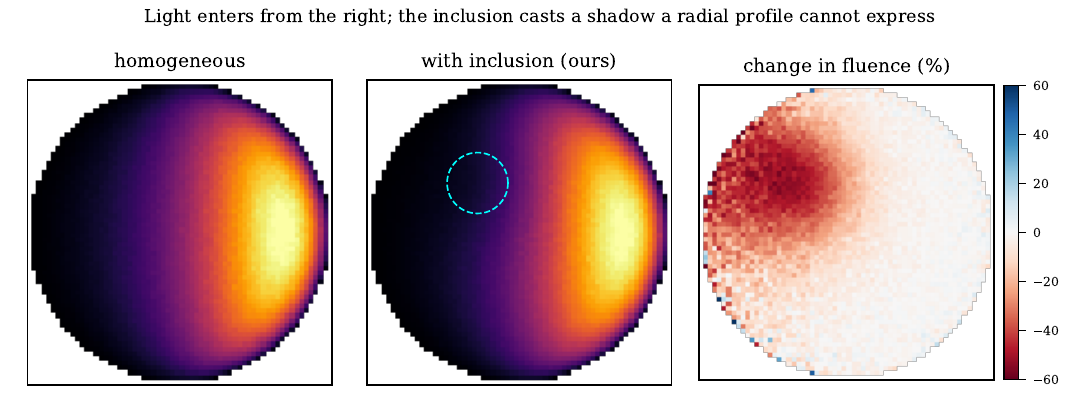}
\caption{Spatially varying media, which a single $R_d(r)$ cannot represent.
Internal diffuse fluence of a translucent disk lit from the right, computed
grid-free, without (left) and with (middle) a denser inclusion (dashed circle).
The inclusion casts a shadow that drops the fluence behind it by $\RShadowDrop\%$
(right); the bulk is unchanged. Solved with the null-scattering
variable-coefficient estimator validated against a manufactured solution.}
\label{fig:hetero}
\end{figure}

\section{Recovering hidden structure}
\label{sec:inverse}
Because the forward problem is a diffusion solve, the corresponding inverse problem is
also accessible, and it is a familiar one: recovering an absorbing inclusion from
diffuse boundary or interior measurements is diffuse optical
tomography~\cite{arridge1999}, here posed grid-free on a signed-distance domain. A
profile-based model cannot pose it at all, having no representation of internal
structure. Our estimator solves the coefficient-to-measurement map and is
differentiable~\cite{miller2024diff,yilmazer2024}, so the recovery experiment is
reported together with the sensitivity and resolution limits that determine when
it is identifiable.

\paragraph{Recovery.}
We fix the walk seed so the Monte Carlo forward model is a deterministic, smooth
function of the inclusion parameters (common random numbers). From internal
fluence probes under $\RInvSrc$ illuminations, with $2\%$ measurement noise, we
recover the centre and contrast of a Gaussian inclusion (a
$\RInvContrastTrue\times$ rise in absorption) by Adam descent on the data loss
from a neutral start at the object's centre (Figure~\ref{fig:inverse}a). The
recovered centre lands within $\RInvCentreErr\%$ of the object's radius and the
contrast within $\RInvContrastErr\%$.

\paragraph{Sensitivity and resolution.}
What can be recovered is governed by the Frechet derivative of the
measurement~\cite{arridge1999}. Linearizing the screened-Poisson equation about a
contrast perturbation $\delta c$ gives the Born kernel
\begin{equation}
\delta m(\mathbf{x}_d)=-\!\int_\Omega G(\mathbf{x}_d,\mathbf{x})\,\phi(\mathbf{x})\,
\delta c(\mathbf{x})\,d\mathbf{x},
\label{eq:born}
\end{equation}
the product of the forward field $\phi$ and the detector Green's function $G$,
both of which the screened-Poisson walks evaluate grid-free. Both decay as
$e^{-\kappa r}$, so the signal of an
inclusion at depth $\delta$ below the measurement surface falls as
$e^{-\kappa\delta}$. We measure exactly this: the perturbation a fixed inclusion
makes in a near-boundary ring decays at $\RInvDecay$ per diffusion length
(Figure~\ref{fig:inverse}b), the single-pass rate \eqref{eq:born} predicts.
Against a noise floor this sets the resolution: an inclusion is recoverable only
while its signal exceeds the noise, i.e. to a depth
$\delta_{\max}\approx\kappa^{-1}\ln(S_0/\eta)$, which is $\RInvResTwo\,\ld$ at
$2\%$ noise and $\RInvResFive\,\ld$ at $5\%$ (Figure~\ref{fig:inverse}b). The
recovery above succeeds because its interior probes sit close to the inclusion,
where the signal is strong. The boundary-only limit is the harder, realistic one,
and the diffusion sets it, not the solver. The forward map is strongly
smoothing, so recovering a low-dimensional inclusion is well posed while
recovering arbitrary internal structure from boundary data is the ill-posed
problem the resolution bound describes. Differentiability enables the low-dimensional recovery but not the general one. The grid-free formulation adds one thing: the forward
field and the Green's function that determine this sensitivity are evaluated on
the true geometry by the same walks, with no mesh.

\begin{figure}[!tbp]
\centering
\includegraphics[width=\linewidth]{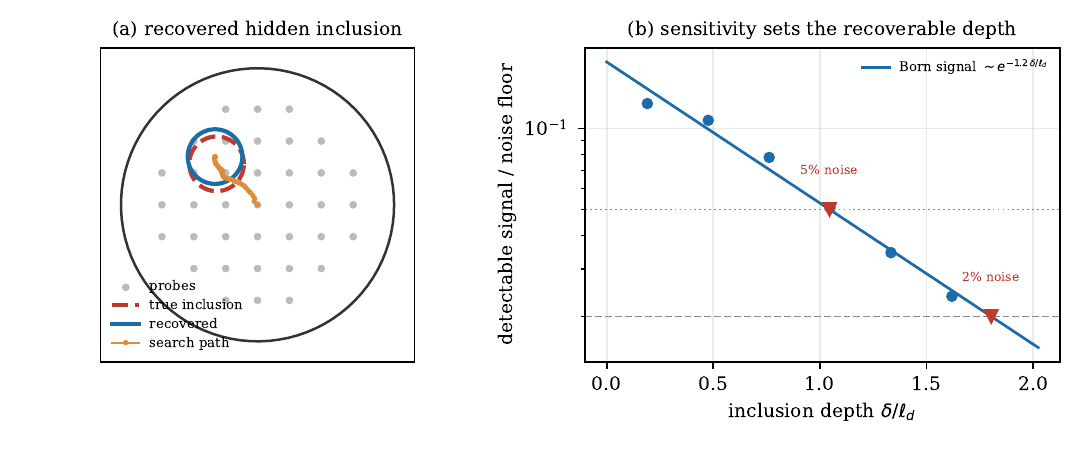}
\caption{Inverse subsurface scattering as grid-free diffuse optical tomography.
\textbf{(a)} From internal diffuse measurements of a translucent disk under
$\RInvSrc$ illuminations, the search path (orange) walks the estimate from a
neutral start onto the true inclusion (dashed); the recovered circle (blue) lands
within $\RInvCentreErr\%$ of the object radius. \textbf{(b)} The Born signal of a
fixed inclusion (circles) decays with depth at the single-pass rate
\eqref{eq:born} (line); where it meets a measurement-noise floor (grey)
sets the recoverable depth (red markers). A radial reflectance profile cannot
represent, let alone recover, this structure.}
\label{fig:inverse}
\end{figure}

\section{Discussion and limitations}
\label{sec:disc2}
The estimator inherits the diffusion approximation it solves. Figure~\ref{fig:discovery}b
shows the geometry-aware diffusion and real transport diverging slightly for
$\tau\lesssim0.4$: when the object is much thinner than a diffusion length, low
orders of scattering dominate and no diffusion model (dipole, multipole or
ours) is exact. Our contribution is to remove the \emph{geometry} error within
the diffusion regime, not to replace diffusion by transport; for the very thin or
optically thin regime a brute-force or higher-order method remains necessary, and
the same distance-field walk supports the brute-force tracer we use here. The two
errors should be kept apart: against the back-lit path-traced reference of
Section~\ref{sec:groundtruth} the remaining $\RGtErrOurs\%$ raw error is
dominated by this \emph{model} error, the gap between diffusion and transport,
since the \emph{solver} error against the closed-form ball is only
$\RMaxRelErr\%$ (Figure~\ref{fig:validation}). Across the \RGtCases-case
benchmark the grid-free estimator has mean raw error $\RGtMeanOurs\%$, close to
the corrected profile at $\RGtMeanCorr\%$, and it does not win every
case. ``Geometry-exact'' refers to the screened-Poisson problem, not to
radiative transfer.

\paragraph{Thickness sets the exponent; curvature enters the prefactor.}
Proposition~\ref{prop:law} is a slab theorem: thickness is the only geometric
degree of freedom there. It does not follow that the pointwise error on a curved,
non-convex shape is governed by the local thickness alone. To measure the extra
geometric dependence we rendered six shapes, two spheres, a fat and a thin torus, a lens and
a ring-and-bead, under uniform illumination (which removes directional shading and
leaves the geometry error), and regressed the log of the dipole's per-point error
against the local thickness and against the local thickness plus the mean
curvature $\ld H(\mathbf{x})$ measured from the distance-field Laplacian
(Figure~\ref{fig:geomsweep}). Thickness alone explains only $R^2=\RGswFitTau$ of
the variance; adding $\ld H$ raises it to $\RGswFitCurv$, with a positive
coefficient $\RGswCoefCurv$. At a fixed $\tau$ the error climbs with curvature: a
strongly curved patch is much worse than a flat one of the same thickness. The
exponent $e^{-2\tau}$ is the slab law and survives, but the prefactor is not
constant; it carries the curvature, so on general geometry the error behaves as
$C(\ld H)\,e^{-2\tau}$ rather than $C\,e^{-2\tau}$. Thickness governs the thin,
nearly flat features the predictor is built for, where $\ld H$ is small.

This form follows from the operator, not just the regression. Near a boundary of mean curvature $H$, written in
normal coordinates with depth $z$, the diffusion operator is
$\partial_z^2 - 2H\,\partial_z - c$. This is the flat half-space operator
$\partial_z^2 - c$ that the dipole solves, plus a first-order drift
$-2H\,\partial_z$. Treating that drift as a perturbation makes the
Green's-function discrepancy \eqref{eq:greendiff}, and with it the prefactor,
depend on curvature at first order through the single dimensionless group
$\ld H$, $\,C(\ld H)=C_0\,(1 + c_1\,\ld H + O((\ld H)^2))$, with $\ld$ the only
length available to make the curvature dimensionless. The dipole's error is the
size of this deviation from the flat value, so it grows with $\lvert\ld
H\rvert$, which is what the regression measures. The derivation therefore fixes
the shape of the correction and the variable it depends on instead of assuming
them. What stays open is the coefficient $c_1$ in closed form for a general
boundary, and a bound uniform in $\tau$: the leading curvature integral of the
half-space Green's function is small, while the rendered error samples
curvatures large enough that higher-order terms enter, so a pointwise bound on
\eqref{eq:greendiff} carrying both optical distance and boundary curvature
remains open.

\begin{figure}[!tbp]
\centering
\includegraphics[width=\linewidth]{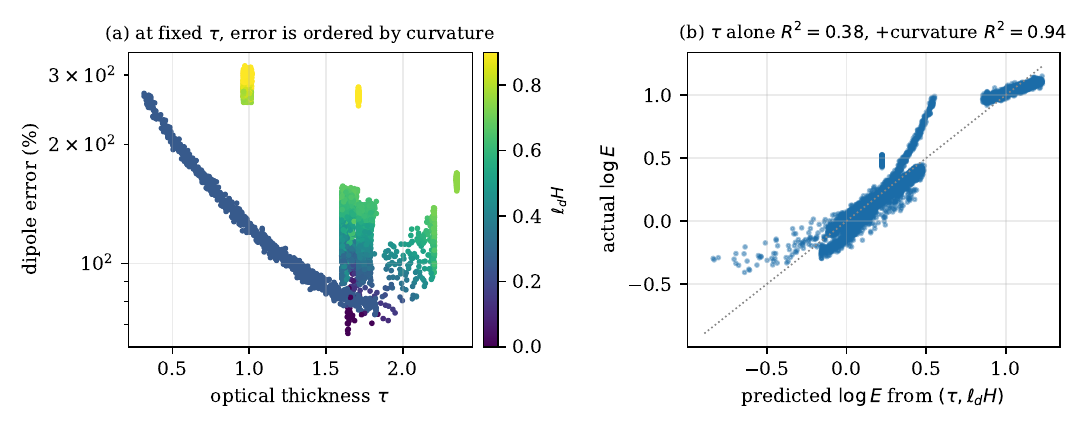}
\caption{Does thickness alone govern the error on general geometry? \textbf{(a)}
The dipole's per-point error against the grid-free reference, over $\RGswShapes$
shapes under uniform illumination, versus optical thickness $\tau$, coloured by
mean curvature $\ld H$. At a fixed $\tau$ the error is ordered by curvature.
\textbf{(b)} Regressing $\log E$ on $\tau$ alone reaches $R^2=\RGswFitTau$; adding
$\ld H$ reaches $\RGswFitCurv$. The exponent is the slab law; the prefactor carries
the curvature.}
\label{fig:geomsweep}
\end{figure}

We adopt the extrapolated-boundary (Dirichlet) condition because it is the one the
dipole assumes, which isolates geometry. It is also the standard linearization of
the physical partial-current (Robin) condition, so it is the textbook choice in
the diffusive regime this paper targets. The exact Robin condition is available
grid-free through Walk on Stars~\cite{sawhney2023} and Walkin'
Robin~\cite{miller2024}; pairing their reflecting walks with the screened source
estimator here is a direct extension that would remove the linearization, which
matters most for high-index or low-albedo media where the partial current
departs furthest from its linearized form.

Two costs bound the reach of the estimator, and we measure both
(Figure~\ref{fig:scope}). Spatially varying media are handled by a null-scattering
majorant~\cite{sawhney2022,sawhney2025sota}, whose variance climbs with the
contrast of the medium. Behind a vein, the relative noise at a fixed budget rises
from $\RHstNoiseMin\%$ at a mild $2\times$ contrast to $\RHstNoiseMax\%$ at
$\RHstRatioMax\times$, growing as $r^{\RHstExpG}$, and the per-ball majorant lowers it
by a constant factor without changing the trend; the $2\times$ manufactured test is
therefore optimistic, and very dense veins remain expensive. The geometry must
also be a true distance. The walk is unbiased for any conservative under-estimate,
but a signed-distance field that is loose by a factor $\alpha$, as voxel and
neural fields often are, lengthens each walk as $\alpha^{-\RHstSdfExp}$, so the
analytic shapes used here are a best case for throughput and detailed production
assets cost more. Single scattering is added separately by the ray-march of
Section~\ref{sec:single}, as is standard.

\begin{figure}[!tbp]
\centering
\includegraphics[width=\linewidth]{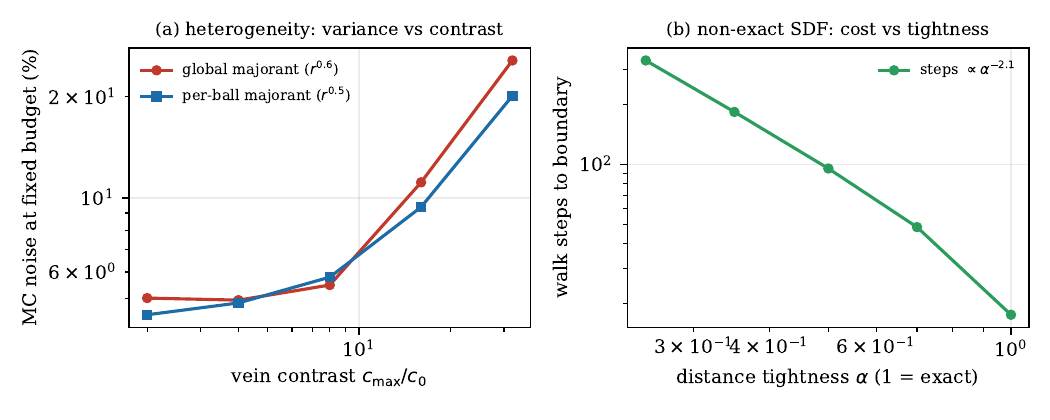}
\caption{Two measured limits of the solver. \textbf{(a)} For a vein of growing
contrast, the Monte Carlo noise at a fixed walk budget rises as
$r^{\RHstExpG}$; the per-ball majorant reduces it by a constant factor. \textbf{(b)}
When the signed-distance field under-estimates distance by a factor $\alpha$ (a
non-exact field), the walk stays unbiased but its step count grows as
$\alpha^{-\RHstSdfExp}$.}
\label{fig:scope}
\end{figure}

The cost is per-pixel Monte Carlo, which suits offline and
interactive-preview rendering and validating faster approximations. Its
variance is reduced by the usual means (more walks, control variates, denoising)
rather than by precomputation. Our evaluation is deliberately mechanism-focused:
it isolates the geometry error against the dipole and a path-traced reference,
rather than benchmarking against the full modern pipeline (normalized
diffusion, separable screen-space, neural and reservoir methods) across many
assets, lightings, and temporal metrics. We now include a small tone-mapped image
metric in Section~\ref{sec:groundtruth}, but not a full perceptual or temporal
study. Establishing where this
estimator is the better choice in a production setting is a separate, empirical
study that the present mathematical results are meant to motivate.

\section{Conclusion}
Subsurface models built on a radial profile carry a hidden assumption of a flat,
bottomless slab. We have shown that for the thin features where it matters most
the error from that assumption is set by the optical thickness
$\tau=T/\ld$, with the dipole over-glowing them by tens of percent until
$\tau\gtrsim\RTauOne$; on strongly curved geometry the boundary
curvature enters the prefactor (Figure~\ref{fig:geomsweep}). The
error is present in radiative transport as well as in the diffusion comparison,
and it is structured across thin translucent features. Solving the diffusion
equation directly inside the object with a Walk-on-Spheres estimator on a
signed-distance function removes the half-space geometry error within that
diffusion model, without meshing the interior, and matches closed forms to
$\RMaxRelErr\%$. The path-traced benchmark also shows the boundary of the claim:
front-lit single scattering, angular transport and curvature can dominate the
remaining image error. The same walks that render the object also give a useful,
but not general, thin-feature predictor and, when differentiated, support
low-dimensional recovery of hidden absorbers. The controlling geometric quantity
in the slab limit is the optical distance light must travel before it reaches a
second surface, not thickness in scene units.

\paragraph{Reproducibility.}
All code (the Metal kernels, the Swift driver and the Python analysis), the
closed-form references, and the scripts that regenerate every figure and number
are included as ancillary files.

\bibliographystyle{plain}
\bibliography{refs}

@inproceedings{jensen2001,
  author    = {Jensen, Henrik Wann and Marschner, Stephen R. and Levoy, Marc and Hanrahan, Pat},
  title     = {A Practical Model for Subsurface Light Transport},
  booktitle = {Proceedings of SIGGRAPH 2001},
  pages     = {511--518},
  year      = {2001},
  publisher = {ACM}
}

@inproceedings{donner2005,
  author    = {Donner, Craig and Jensen, Henrik Wann},
  title     = {Light Diffusion in Multi-Layered Translucent Materials},
  booktitle = {ACM SIGGRAPH 2005 Papers},
  pages     = {1032--1039},
  year      = {2005},
  publisher = {ACM}
}

@article{deon2011,
  author  = {d'Eon, Eugene and Irving, Geoffrey},
  title   = {A Quantized-Diffusion Model for Rendering Translucent Materials},
  journal = {ACM Transactions on Graphics},
  volume  = {30},
  number  = {4},
  pages   = {56:1--56:14},
  year    = {2011}
}

@article{habel2013,
  author  = {Habel, Ralf and Christensen, Per H. and Jarosz, Wojciech},
  title   = {Photon Beam Diffusion: A Hybrid Monte Carlo Method for Subsurface Scattering},
  journal = {Computer Graphics Forum (Proc. EGSR)},
  volume  = {32},
  number  = {4},
  pages   = {27--37},
  year    = {2013}
}

@article{frisvad2014,
  author  = {Frisvad, Jeppe Revall and Hachisuka, Toshiya and Kjeldsen, Thomas Kim},
  title   = {Directional Dipole Model for Subsurface Scattering},
  journal = {ACM Transactions on Graphics},
  volume  = {34},
  number  = {1},
  pages   = {5:1--5:12},
  year    = {2014}
}

@techreport{christensen2015,
  author      = {Christensen, Per H. and Burley, Brent},
  title       = {Approximate Reflectance Profiles for Efficient Subsurface Scattering},
  institution = {Pixar Animation Studios},
  number      = {15-04},
  year        = {2015}
}

@inproceedings{burley2015,
  author    = {Burley, Brent},
  title     = {Extending the Disney {BRDF} to a {BSDF} with Integrated Subsurface Scattering},
  booktitle = {ACM SIGGRAPH 2015 Courses: Physically Based Shading in Theory and Practice},
  year      = {2015},
  publisher = {ACM}
}

@inproceedings{stam1995,
  author    = {Stam, Jos},
  title     = {Multiple Scattering as a Diffusion Process},
  booktitle = {Rendering Techniques (Proc. Eurographics Workshop on Rendering)},
  pages     = {41--50},
  year      = {1995}
}

@article{sawhney2020,
  author  = {Sawhney, Rohan and Crane, Keenan},
  title   = {Monte Carlo Geometry Processing: A Grid-Free Approach to {PDE}-Based Methods on Volumetric Domains},
  journal = {ACM Transactions on Graphics},
  volume  = {39},
  number  = {4},
  pages   = {123:1--123:18},
  year    = {2020}
}

@article{sawhney2022,
  author  = {Sawhney, Rohan and Seyb, Dario and Jarosz, Wojciech and Crane, Keenan},
  title   = {Grid-Free Monte Carlo for {PDE}s with Spatially Varying Coefficients},
  journal = {ACM Transactions on Graphics},
  volume  = {41},
  number  = {4},
  pages   = {53:1--53:17},
  year    = {2022}
}

@article{sawhney2023,
  author  = {Sawhney, Rohan and Miller, Bailey and Gkioulekas, Ioannis and Crane, Keenan},
  title   = {Walk on Stars: A Grid-Free Monte Carlo Method for {PDE}s with Neumann Boundary Conditions},
  journal = {ACM Transactions on Graphics},
  volume  = {42},
  number  = {4},
  pages   = {80:1--80:20},
  year    = {2023}
}

@article{miller2024,
  author  = {Miller, Bailey and Sawhney, Rohan and Crane, Keenan and Gkioulekas, Ioannis},
  title   = {Walkin' Robin: Walk on Stars with Robin Boundary Conditions},
  journal = {ACM Transactions on Graphics},
  volume  = {43},
  number  = {4},
  pages   = {41:1--41:18},
  year    = {2024}
}

@book{ishimaru1978,
  author    = {Ishimaru, Akira},
  title     = {Wave Propagation and Scattering in Random Media},
  publisher = {Academic Press},
  year      = {1978}
}

@article{vicini2019,
  author  = {Vicini, Delio and Koltun, Vladlen and Jakob, Wenzel},
  title   = {A Learned Shape-Adaptive Subsurface Scattering Model},
  journal = {ACM Transactions on Graphics},
  volume  = {38},
  number  = {4},
  pages   = {127:1--127:15},
  year    = {2019}
}

@article{neupress2024,
  author  = {Tg, T. and Frisvad, Jeppe Revall and Ramamoorthi, Ravi and Jensen, Henrik Wann},
  title   = {{NeuPreSS}: Compact Neural Precomputed Subsurface Scattering for Distant Lighting of Heterogeneous Translucent Objects},
  journal = {Computer Graphics Forum},
  volume  = {43},
  number  = {7},
  pages   = {e15234},
  year    = {2024}
}

@article{werner2024,
  author  = {Werner, Mirco and Sch{\"u}{\ss}ler, Vincent and Dachsbacher, Carsten},
  title   = {{ReSTIR} Subsurface Scattering for Real-Time Path Tracing},
  journal = {Proceedings of the ACM on Computer Graphics and Interactive Techniques},
  volume  = {7},
  number  = {3},
  pages   = {1--19},
  year    = {2024}
}

@article{miller2024diff,
  author  = {Miller, Bailey and Sawhney, Rohan and Crane, Keenan and Gkioulekas, Ioannis},
  title   = {Differential Walk on Spheres},
  journal = {ACM Transactions on Graphics},
  volume  = {43},
  number  = {6},
  pages   = {269:1--269:18},
  year    = {2024}
}

@article{yilmazer2024,
  author  = {Yilmazer, Ekrem Fatih and Vicini, Delio and Jakob, Wenzel},
  title   = {Solving Inverse {PDE} Problems using Grid-Free Monte Carlo Estimators},
  journal = {ACM Transactions on Graphics (Proc. SIGGRAPH Asia)},
  volume  = {43},
  number  = {6},
  pages   = {270:1--270:18},
  year    = {2024}
}

@article{shapira2008,
  author  = {Shapira, Lior and Shamir, Ariel and Cohen-Or, Daniel},
  title   = {Consistent Mesh Partitioning and Skeletonisation using the Shape Diameter Function},
  journal = {The Visual Computer},
  volume  = {24},
  number  = {4},
  pages   = {249--259},
  year    = {2008}
}

@inproceedings{walter2007,
  author    = {Walter, Bruce and Marschner, Stephen R. and Li, Hongsong and Torrance, Kenneth E.},
  title     = {Microfacet Models for Refraction through Rough Surfaces},
  booktitle = {Proceedings of the Eurographics Symposium on Rendering},
  pages     = {195--206},
  year      = {2007}
}

@inproceedings{sawhney2025sota,
  author    = {Sawhney, Rohan and Miller, Bailey and Gkioulekas, Ioannis and Crane, Keenan and Jarosz, Wojciech and Zhao, Shuang and Nabizadeh, Mohammad Sina and Li, Zilu},
  title     = {State of the Art in Grid-Free Monte Carlo Methods for Partial Differential Equations},
  booktitle = {ACM SIGGRAPH 2025 Courses},
  year      = {2025},
  doi       = {10.1145/3721241.3734001}
}

@article{muller1956,
  author  = {Muller, Mervin E.},
  title   = {Some Continuous Monte Carlo Methods for the Dirichlet Problem},
  journal = {The Annals of Mathematical Statistics},
  volume  = {27},
  number  = {3},
  pages   = {569--589},
  year    = {1956}
}

@book{agmon1982,
  author    = {Agmon, Shmuel},
  title     = {Lectures on Exponential Decay of Solutions of Second-Order Elliptic Equations: Bounds on Eigenfunctions of {$N$}-Body Schr\"odinger Operators},
  series    = {Mathematical Notes},
  volume    = {29},
  publisher = {Princeton University Press},
  year      = {1982}
}

@book{oksendal2003,
  author    = {{\O}ksendal, Bernt},
  title     = {Stochastic Differential Equations: An Introduction with Applications},
  edition   = {6},
  publisher = {Springer},
  year      = {2003}
}

@article{arridge1999,
  author  = {Arridge, Simon R.},
  title   = {Optical Tomography in Medical Imaging},
  journal = {Inverse Problems},
  volume  = {15},
  number  = {2},
  pages   = {R41--R93},
  year    = {1999}
}

@inproceedings{chiang2016,
  author    = {Chiang, Matt Jen-Yuan and Kutz, Peter and Burley, Brent},
  title     = {Practical and Controllable Subsurface Scattering for Production Path Tracing},
  booktitle = {ACM SIGGRAPH 2016 Talks},
  pages     = {49:1--49:2},
  year      = {2016}
}

@article{jimenez2015separable,
  author  = {Jimenez, Jorge and Zsolnai, K\'aroly and Jarabo, Adri\'an and Freude, Christian and Auzinger, Thomas and Wu, Xian-Chun and von der Pahlen, Javier and Wimmer, Michael and Gutierrez, Diego},
  title   = {Separable Subsurface Scattering},
  journal = {Computer Graphics Forum},
  volume  = {34},
  number  = {6},
  pages   = {188--197},
  year    = {2015}
}

\end{document}